\documentclass[a4paper]{paper}
\usepackage{geometry}
\usepackage{url}

\usepackage{cite}
\usepackage{chngcntr}
\counterwithout{figure}{section}

\usepackage{amsmath}
\usepackage{amsthm}
\usepackage{amssymb}
\usepackage{todonotes}

\newtheorem{theorem}{Theorem}
\newtheorem{lemma}{Lemma}
\newtheorem{definition}{Definition}

\newtheorem{claim}{Claim}

\newtheorem{observation}{Observation}

\newcommand{\vecinv}[2]{#1^{-1}(#2)}
\DeclareMathOperator{\symdiff}{\bigtriangleup}

\newtheorem*{rep@theorem}{\rep@title}
\newcommand{\newreptheorem}[2]{%
\newenvironment{rep#1}[1]{%
 \def\rep@title{#2 \ref{##1}}%
 \begin{rep@theorem}[restated]}%
 {\end{rep@theorem}}}
\newreptheorem{theorem}{Theorem}
\newreptheorem{lemma}{Lemma}
\newreptheorem{proposition}{Proposition}

\begin{document}
\title{Sharper Upper Bounds for Unbalanced Uniquely Decodable Code Pairs\footnote{Full version of an extended abstract to appear at the 2016 IEEE International Symposium on Information Theory.}}

\author{
  Per Austrin
  \thanks{School of Computer Science and Communication, KTH Royal Institute of Technology, Sweden. \protect\url{austrin@csc.kth.se}}
  \and
  Petteri Kaski
  \thanks{Helsinki Institute for Information Technology HIIT,  Department of Computer Science, Aalto University, Finland. \protect\url{petteri.kaski@aalto.fi}}
  \and
  Mikko Koivisto
  \thanks{Helsinki Institute for Information Technology HIIT,  Department of Computer Science, University of Helsinki, Finland. \protect\url{mikko.koivisto@helsinki.fi}}
  \and
  Jesper Nederlof
  \thanks{Department of Mathematics and Computer Science, Technical University of Eindhoven, The Netherlands. \protect\url{j.nederlof@tue.nl}}
}

\maketitle

\begin{abstract}
Two sets $A, B \subseteq \{0, 1\}^n$ form a Uniquely Decodable Code Pair
(UDCP) if every pair $a \in A$, $b \in B$ yields a distinct sum $a+b$, where the addition is over $\mathbb{Z}^n$. We show that every UDCP $A, B$, with $|A| = 2^{(1-\epsilon)n}$ and $|B| = 2^{\beta n}$, satisfies $\beta \leq 0.4228 +\sqrt{\epsilon}$.
For sufficiently small $\epsilon$, this bound significantly improves previous bounds by Urbanke and Li~[Information Theory Workshop '98] and Ordentlich and Shayevitz~[2014, arXiv:1412.8415], which upper bound $\beta$ by $0.4921$ and $0.4798$, respectively, as  $\epsilon$ approaches $0$.
\end{abstract}

\section{Introduction}

% 1)The first paragraph of the introduction is not very good. Multiple
% access channels are a well-studied model in the information-theory
% literature and need no further motivation. More importantly, the
% classic notion of Shannon capacity allows for arbitrary small, but
% positive, error probability. Under this notion, the capacity of the
% adder channel is well-known to consists of all pairs
% $R_1<=1,R_2<=1,R_1+R_2<=1.5$. The *zero-error* capacity of the channel
% is the subject of the current paper. Currently, the paper makes no
% distinction between these two notions.  I also think figure 1 adds
% nothing to the understanding of the problem and should be removed.

A canonical problem in multi-user communication theory is how to coordinate unambiguous communication through a channel, such that several independent senders can simultaneously send as much information as possible to a single receiver (see, e.g., the book by Schleger and Grant~\cite{schleger}); this could for example occur when several satellites need to send their data to a single terminal. % While trivial frequency or time multiplexing directly handles such a situation, it might be far from an optimal solution if the channel capacity is a scarce resource: in many settings, multi-user coding schemes can achieve a significantly higher total rate of transmission (sum-rate) than traditional channel multiplexing.

% \begin{figure}[!t]
% \centering
% \includegraphics[width=0.95\linewidth]{picture}
% \caption{The two-user Binary Adder Channel. Two senders independently send $n$-bit strings $a$ and $b$ to a receiver. The messages are encoded into radiowaves that add through a noiseless channel, so the receiver gets $a+b$. To ensure the message is unambiguous for the receiver, the senders need to use messages from a cleverly designed set of allowed messages $A,B\subseteq \{0,1\}^n$.}
% \label{fig:UDCP-illustation}
% \end{figure}

Unfortunately, despite vast research in the last decades, even in some of the simplest models the exact capacity of such communication channels remains far from clear. An extensively investigated and fundamental example is the \emph{two-user Binary Adder Channel (BAC)}.  The zero-error capacity of the BAC is equal to the maximum size of the product of the code sizes of a \emph{Uniquely Decodable Code Pair (UDCP)}: a pair $A,B \subseteq \{0,1\}^n$ such that $|A+B|=|A|\cdot|B|$ where $A+B$ denotes the sumset $\{a+b: a\in A, b \in B\}$, and $a+b$ denotes addition over $\mathbb{Z}^n$.

Most previous research on UDCPs has focused on constructions. A basic observation is that, if $A_1,B_1 \subseteq 2^{[n]}$ is a UDCP\footnote{In this work, we freely interchange vectors with sets in the natural way.} and $A_2,B_2 \subseteq 2^{[n]}$ is a UDCP, then $A_1 \times A_2,B_1\times B_2$ is also a UDCP. Therefore, for finding asymptotically good constructions for every $n$, it is sufficient to focus on finite $n$. Letting $\alpha$ and $\beta$ denote respectively $\log_2(|A|)/n$ and $\log_2(|B|)/n$, a natural and popular goal is to find a UDCP maximizing $\alpha+\beta$. The first and simplest construction, $A=\{00,01,11\},B=\{10,01\}$ giving $\alpha+\beta=(\log_2(3)+1)/2 \approx 1.29248$, was presented by Kasami and Lin~\cite{1055529}. This was the best until 1985. Then it was improved to $1.30366$ by van den Braak and van Tilborg~\cite{1057004}, and after subsequent improvements by Ahlswede and Balakirsky~\cite{746834} ($1.30369$), van den Braak~\cite{braak} ($1.30565$),  Urbanke and Li~\cite{urbankeli} ($1.30999$), the current record is $1.31781$ by Mattas and {\"O}sterg{\aa}rd~\cite{mattas}. Several of these results were obtain by computer searches for finite $n$. More relevant to our study is the important work by Kasami et al.~\cite{Kasami}, which shows that for sufficiently large $n$ there exist (somewhat surprisingly) UDCPs with $\alpha \geq 1 - o(1)$ and $\beta \geq 0.25$.

Considering upper bounds, the rather direct $\alpha + \beta \leq 1.5$ has been independently found by at least Liao~\cite{liao}, Ahlswede~\cite{ahlswede1973multi}, Lindstr\"{o}m~\cite{Lindstrm1972} and van Tilborg~\cite{tilborg1}. Somewhat unsatisfactory, $1.5$ is, to the best of our knowledge, still the best known upper bound on $\alpha+\beta$ in general. However, Urbanke and Li~\cite{urbankeli} managed to break through the $1.5$ bound in the \emph{unbalanced case}: assuming $\alpha \geq 1- \epsilon$ for a sufficiently small value of $\epsilon$, they showed that $\beta \leq 0.4921$. On a high level, their approach works as follows: a result of van Tilborg~\cite{tilborg1} (see also Lemma~\ref{lem:vantilborg} below) shows there are not many pairs $(a,b) \in A \times B$ of small Hamming distance, and if $A$ and $B$ are sufficiently large, then the number of such pairs is lower bounded by an \emph{isoperimetric inequality} for which the authors use Harper's theorem. Later, this result was improved to $\beta \leq 0.4798$ by Ordentlich and Shayevitz~\cite{DBLP:journals/corr/OrdentlichS14b}. Their proof idea is somewhat more involved: the authors give a procedure that, given a UDCP $A,B\subseteq \{0,1\}^n$, constructs another UDCP $C,C \in \{0,1\}^{(1-\gamma)n}$ with some $\gamma>0$. This was achieved by proving the existence of a subset $L \subseteq [n]$ with $|L|=\gamma n$ such that for some $c \in \{0,1,2\}^{|L|}$, the projection $(a+b)_L$ equals $c$ for many pairs $a,b$. The existence of such a subset is proved using a variant of the Sauer--Perles--Shelah lemma. Unfortunately, both the referred bounds \cite{DBLP:journals/corr/OrdentlichS14b,urbankeli} converge fast to $(1-\epsilon)+\beta \leq 1.5$ as $\epsilon$ increases (see Figure~1 of Ordentlich and Shayevitz~\cite{DBLP:journals/corr/OrdentlichS14b}). 

The present authors~\cite{DBLP:conf/stacs/AustrinKKN15} gave a novel and direct connection between UDCPs and additive number theory. Motivated by algorithm design for the Subset Sum problem, they observed the following: if $w \in \mathbb{Z}^n,t \in \mathbb{Z}$ and $A \subseteq \{0,1\}^n$ such that $a \cdot w = a' \cdot w$ implies $a=a'$ for every $a,a'\in A$, and $B = \{b \in \{0,1\}^n: w \cdot b = t\}$, then $A,B$ is a UDCP. Here `$\cdot$' denotes the inner product.

%Let us mention here that 

The channel capacity application has also inspired studies of several variants of the basic setting of this paper, 
for example, with both sets being the same \cite{Lindstrm1972,Cohen:2001:BB:374604.374616}, with noise~\cite{schleger}, or with more than two users~\cite{1056109,485711,ahlswede1973multi}.

\subsection*{Our Contribution}

Motivated by the unsatisfactory slow progress on the large gap between the current lower and upper bounds for UDCPs, we propose to restrict attention to the case $|A| \geq 2^{(1-\epsilon)n}$ for small values of $\epsilon$: before we can understand the exact tradeoff between $\alpha$ and $\beta$, we first need to understand this tradeoff for large values of $\alpha$. An intriguing question is whether $\alpha \ge 1 - o(1)$ implies $\beta \le 0.25 + o(1)$; in other words, is the construction of Kasami et al.~\cite{Kasami} optimal, or could it be improved? We make significant progress on this question, and our main result is:
\begin{theorem}[Main Theorem]
  \label{thm:main}
  If $A,B \subseteq \{0,1\}^n$ is a UDCP with $|A| \geq
  2^{(1-\epsilon)n}$ and $|B| = 2^{\beta n}$, then $\beta \leq 0.4228
  +\sqrt{\epsilon}$.
\end{theorem}

Our proof combines ideas from both previous upper bounds with new ideas. We will present our proof by first providing a ``warm-up'' bound of $\beta \le 0.4777 + O(\sqrt{\epsilon})$ (Theorem~\ref{thm:simplebound}). To establish this bound, we study the joint probability $\Pr[a \in A, b \in B]$ for two \emph{correlated} random strings $a, b \in \{0,1\}^n$. We upper and lower bound this probability using, respectively, van Tilborg's lemma (Lemma~\ref{lem:vantilborg}) and an isoperimetric inequality due to Mossel et al.~\cite{MORSS06-NICD}.  This approach is similar to that of Urbanke and Li \cite{urbankeli}, but improves their bound for small values of $\epsilon$.

The intuition behind our main bound (and, partially, the bounds of Urbanke and Li~\cite{urbankeli} and Ordentlich and Shayevitz~\cite{DBLP:journals/corr/OrdentlichS14b}) is as follows. The above strategy does not give a good bound if $A$ and $B$ are antipodal Hamming balls: the studied probability is very small in this case, so the upper bound is not really stringent. However, intuitively such a pair cannot form a large UDCP since the pairwise sums will be concentrated on the sum of the two centers of the Hamming balls. Our novel approach is that we use the encoding argument from van Tilborg's lemma to show that if $A$ is large enough, then $B$ needs to be sufficiently spread out over the hypercube. Specifically, we show that there exists a set $L \subseteq [n]$ of size close to $n/2$ such that $L$ has an almost maximum number of projections on $B$. Subsequently, we use this set $L$ to define a refined distribution of the strings $x$ and $y$. In the refined distribution, $x,y$ are only correlated in the coordinates from $L$, and for applying the isoperimetric inequality the large number of projections is then essential.

% In contrast to using the Sauer--Perles--Shelah as in~\cite{DBLP:journals/corr/OrdentlichS14b}, to show that such a set $L$ exists, we use the UDCP property in a novel way.   

\section{Notation and Preliminaries}

\subsection{Notation}

%In this work, we will interpret elements $\{0,1\}^{n}$ as either vectors or subsets of $[n]$ in which ever way is convenient.

Given reals $a,b$ with $b \geq 0$, we write $a\pm b$ for the interval
$[a-b,a+b]$. If $n$ is an integer, we denote $[n]=\{1,\ldots,n\}$.
For $x \in \mathbb{R}^n$, we denote by $\vecinv{x}{z} \subseteq [n]$
the set of coordinates $i$ such that $x_i=z$.

For binary vectors, we extend notation for subsets of $[n]$ in the
obvious way (by interpreting $x \in \{0,1\}^n$ as the set
$\vecinv{x}{1} \subseteq [n]$).  Thus e.g.~$x \setminus y$ is a vector
which is $1$ in the coordinates $i$ where $x_i = 1$ and $y_i = 0$,
$x \symdiff y$ denotes the symmetric difference (or alternatively, the
componentwise XOR) of $x$ and $y$, and $|x|$ denotes the Hamming
weight of $x$.

Given $x \in \{0,1\}^n$ and $P \subseteq [n]$, we let $x_P$ denote the
\emph{projection} of $x$ on $P$: $x_P\in \{0,1\}^P$ such that $x_P$
agrees with $x$ on all coordinates in $P$. For a family $X \subseteq
\{0,1\}^n$ we also write $X_P = \{x_P: x \in X\}$.

\subsection{Entropy}

For $x\in [0,1]$ we let $h(x)=-x\log_2(x) -(1-x)\log_2(1-x)$ denote the \emph{binary entropy} of $x$. It is well known that $h(x)$ is monotone increasing for $x \in [0,1/2]$, monotone decreasing for $x\in[1/2,1]$, and that $\binom{n}{t} \leq 2^{h(t/n)n}$. The following elementary inequality can be shown by standard calculus:

\begin{observation}
  \label{obs:binary entropy estimate around half}
  For all $x \in (0,1/2]$, $h\!\left(\tfrac{1}{2} + x\right) < 1 - \frac{2}{\ln 2} x^2$.
\end{observation}

This observation implies another useful bound:

\begin{observation}\label{obs:fatlayer}
Let $\epsilon >0$ be a constant. Suppose $X \subseteq \{0,1\}^n$ such that $|X|\geq 2^{(1-\epsilon)n}$, $z \in \{0,1\}^n$, and $\gamma \ge \sqrt{\frac{\ln 2}{2} \epsilon}$. Then for sufficiently large $n$, we have that $|\{x \in X: |x \symdiff z| \in (\tfrac{1}{2} \pm \gamma)n \}| \geq |X|/2$.
\end{observation}
%To see this, note that $|\{x \in X: |x \symdiff z| \not\in (0.5 \pm \gamma)n \}| \leq \binom{n}{(0.5+\gamma)n}\leq 2^{h(0.5+\gamma)n} \leq 2^{(1-\frac{2}{\ln 2}\gamma^2)n} \leq |X|/2$ for sufficiently large $n$.

\subsection{UDCPs}

We will use the following well known property of UDCPs that directly follows from noting that whenever $a-b=a'-b'$ we have $a+b'=a'+b$:
\begin{observation}\label{obs:minudcp}
If $A,B$ is a UDCP, then $|A - B|=|A|\cdot|B|$.
\end{observation}

We will also use the following bound. Since the proof is elegant and highly instructive for understanding our approach, we provide a (known) proof.
\begin{lemma}[van Tilborg~\cite{tilborg1}]\label{lem:vantilborg}
Let $A,B \subseteq \{0,1\}^n$ be a UDCP and let $W_d=|\{(a,b) \in A \times B: |a \symdiff b| = d \}|$. Then $|W_d| \leq \binom{n}{d}2^{\min\{d,n-d\}}$.
\end{lemma}
\begin{proof}
Let us bound the number of possibilities for $a+b$ and $b-a$ for pairs $(a,b)\in W_d$. Note that
\[
a \symdiff b = \vecinv{(a+b)}{1} = [n] \setminus \vecinv{(b-a)}{0}  \,.
\]
Thus, since $|a \symdiff b|=d$, fixing $a \symdiff b$ (in one of the
$\binom{n}{d}$ possible ways) leaves either $2^{n-d}$ possible choices for
$\vecinv{(a+b)}{0}$ and $\vecinv{(a+b)}{2}$, or $2^d$ possible choices
for $\vecinv{(b-a)}{-1}$ and $\vecinv{(b-a)}{1}$.  By the UDCP
property, either of these two completely determines $(a,b) \in W_d$, and the bound
follows.
\end{proof}

\subsection{$\rho$-correlation and isoperimetry}

For $x \in \{0,1\}^U$, we write $y \sim_\rho x$ for a \emph{$\rho$-correlated random copy} of $x$, i.e., a string where, independently for each $e \in U$, 
\[
y_e = \begin{cases}
  x_e, &\text{with probability } \frac{1+\rho}{2}\,,\\
  1-x_e, &\text{with probability } \frac{1-\rho}{2}\,.
\end{cases}
\]
If $x$ is not fixed, we use $y \sim_\rho x$ to denote the joint
distribution over $(x,y)$ where $x$ is a uniformly random string and
$y$ is $\rho$-correlated copy of $x$. Our bounds will rely on the reverse Small Set Expansion Theorem, an isoperimetric inequality of the noisy Boolean hypercube: 

\begin{lemma}[Reverse Small Set Expansion, {\cite[Theorem~3.4]{MORSS06-NICD}}\footnote{In the notation of~\cite{MORSS06-NICD} where $|F| \ge
    e^{-s^2/2} 2^{|U|}$ and $|G| \ge e^{-t^2/2} 2^{|U|}$ we have $s =
    \sqrt{2 \ln 2 (1-f) |U|}$ and $t = \sqrt{2 \ln 2 (1-g) |U|}$.}]
  \label{thm:rsse}
  Let $F, G \subseteq \{0,1\}^{U}$ with $|F|\geq 2^{f |U|}$, $|G|\geq 2^{g |U|}$. Then
  \[
  \Pr_{y \sim_{\rho} x}[x \in F, y \in G] \geq 2^{-|U|\left(\frac{(1-f)+(1-g)+2\rho\sqrt{(1-f)(1-g)}}{1-\rho^2}\right)}\,.
  \]
\end{lemma}

\section{Simple UDCP Bound Using Isoperimetry}

In this section we give a warm-up to our main result, showing how a
simple application of Theorem~\ref{thm:rsse} suffices to obtain
improved UDCP bounds.

\begin{theorem}\label{thm:simplebound}
  If $A,B \subseteq \{0,1\}^n$ is a UDCP with $|A| \ge
  2^{(1-\epsilon) n}$ and $|B|\geq 2^{\beta n}$, then $\beta \leq
  0.4777+\epsilon +0.7676\sqrt{\epsilon(1-\beta)}$.
\end{theorem}
\begin{proof}
Let $W_d=\{(a,b) \in A \times B: |a \symdiff b| = d \}$.  By definition of $\rho$-correlation it is easy to see that
\begin{align*}
	 \Pr_{a \sim_{\rho} b}[a \in A, b \in B]&= 2^{-n}\sum_{d=0}^{n} \left(\frac{1+\rho}{2}\right)^{n-d} \left(\frac{1-\rho}{2}\right)^d |W_d|\\
																				  &\le 2^{-2n}\sum_{d=0}^{n} (1+\rho)^{n-d} (1-\rho)^d \binom{n}{d}2^{d}\\
																					&= 2^{-2n} (3-\rho)^n\,,
\end{align*}
where the inequality follows from
Lemma~\ref{lem:vantilborg},\footnote{Here we did not use the full
  strength of Lemma~\ref{lem:vantilborg}.  In particular we only use
  that $|W_d| \leq \binom{n}{d}2^d$.  However, using the sharper bound
  of $\binom{n}{d} 2^{\min(d, n-d)}$ does not yield any improvement in
  the exponent because the dominating terms in the exponential sum are
  those where $d \le n/2$.} and the last equality follows from the
Binomial Theorem. On the other hand, using Theorem~\ref{thm:rsse}, we
have that
\[
	\Pr_{a \sim_{\rho} b}[a \in A, b \in B] \geq 2^{-n\left(\frac{\epsilon+(1-\beta)+2\rho\sqrt{\epsilon(1-\beta)}}{1-\rho^2}\right)}\,.
\]
Combining the bounds, taking logs, and dividing by $n$, we see that for any $0\leq \rho < 1$, 
\[
		-\left(\frac{\epsilon+1-\beta+2\rho\sqrt{\epsilon(1-\beta)}}{1-\rho^2}\right) \leq \log_2(3-\rho)-2\,,
\]
or equivalently,
\[
 \beta \le (\log_2(3-\rho)-2)(1-\rho^2) + 1 + \epsilon + 2\rho\sqrt{\epsilon(1-\beta)}\,.
\]
Setting $\rho=0.3838$ we obtain
\[
	\beta \leq 0.4777+\epsilon +0.7676\sqrt{\epsilon(1-\beta)}\,.
\]
\end{proof}

\section{Proof Overview of Main Bound}
\label{sec:refined dist}

The proof of our main bound follows the same blueprint as the proof of Theorem~\ref{thm:simplebound}, but we use a more refined version of the noise distribution. In particular, we only apply the noise on a subset of $[n]$ where both $A$ and $B$ are sufficiently dense.

\begin{definition}
  Fix $L \subseteq [n]$.  Given $x\in\{0,1\}^n$ we let $y\sim^L_{\rho}x$ denote that $y\in \{0,1\}^n$ is the random variable distributed as follows:
\[
	y_i =
    \begin{cases}
      y_i \sim_\rho x_i & \text{if $i \in L$} \\
      y_i \sim_0 x_i & \text{if $i \not\in L$.}
    \end{cases}
\]
(I.e., $y$ is a $\rho$-correlated copy of $x$ on the coordinates of $L$, and uniformly random outside $L$.)
\end{definition}

We proceed to give upper and lower bounds on
the quantity $\Pr_{a \sim^L_{\rho} b}[a \in A, b \in B]$.  In order
for these bounds to hold, we need a mild density condition on $A$
with respect to the split $(L, [n] \setminus L)$.  In particular, we
make the following definition.

\begin{definition}
  We say that $A \subseteq \{0,1\}^n$ is $\epsilon$-dense with respect
  to $L \subseteq [n]$ if $|A_L| \ge 2^{|L|-\epsilon n-1}$, and for
  every $a \in A$, the number of $a' \in A$ such that $a_L = a'_L$ is
  at least $2^{n-|L| - \epsilon n - 1}$.
\end{definition}

As the following simple claim shows, our set $A$ is guaranteed to have
a dense subset.

\begin{claim}
  \label{claim:freq}
  Let $A \subseteq \{0,1\}^n$ such that $|A| \ge 2^{(1-\epsilon)n}$.
  Then for any $L \subseteq [n]$, there is an $A' \subseteq A$ that
  is $\epsilon$-dense with respect to $L$.
\end{claim}

\begin{proof}
  For $a, a' \in A$ note that the condition $a_L = a'_L$ is an
  equivalence relation partitioning $A$ into at most $2^{|L|}$
  equivalence classes, each of size at most $2^{n-|L|}$.  It follows
  that there must be at least $|A|/2^{n-|L|+1} \ge 2^{|L| - \epsilon n
    - 1}$ equivalence classes of size at least $|A|/2^{|L|+1} =
  2^{n-|L| - \epsilon n - 1}$ and we can take $A'$ to be the union of
  these.
\end{proof}

With these definitions in place, we are ready to state the precise
upper and lower bounds on the refined noise probability.

\begin{lemma}
\label{lem:encoding-bound}
Fix $L \subseteq [n]$ and let $\lambda = |L|/n$.  Then for any $0 \le \rho \le 1$ and UDCP
$(A,B)$ such that $|A|$ is $\epsilon$-dense with respect to $L$, we have
\[
\frac{\log_2 \Pr_{a \sim^L_{\rho} b}[a \in A, b \in B]}{n} \leq \sqrt{\tfrac{\ln(2)\epsilon}{2}}-\tfrac{1}{2}+\lambda \cdot \left(\log_2(3-\rho)-\tfrac{3}{2}\right) + o(1).
\]
\end{lemma}

The proof appears in Section~\ref{sec:ub-proof}.

\begin{lemma}\label{lem:lowbnd}
  Fix $L \subseteq [n]$ with $|L| = \lambda n$.  Then for any constant
  $0 \le \rho < 1$ the following holds.  Let $(A,B)$ be a UDCP such
  that $A$ is $\epsilon$-dense with respect to $L$, and $|B_L| =
  2^{\pi n}$ for some $0 \le \pi \le \lambda$. Then
  \[
  \frac{\log_2 \Pr_{a \sim^L_{\rho} b}[a \in A, b \in B]}{n} \ge \frac{\pi-\lambda-\epsilon-2\rho\sqrt{\epsilon(\lambda-\pi)}}{1-\rho^2}+\lambda-1-\epsilon - o(1).
  \]
  The constant in the $o(1)$ term depends on $\lambda, \rho,
  \epsilon$ and $\pi$, and is finite assuming $\epsilon$ is bounded
  away from $0$ and $\rho$ is bounded away from $1$.
\end{lemma}

The proof appears in Section~\ref{sec:lb-proof}.

The quality of the lower bound depends on the size of $|B_L|$ and in
particular we would like to find a split $L$ such that $|B_L| \approx
|B|$.  At the same time we would like $|L|$ to be as small as
possible.  The following Lemma shows that we can take $|L|
\approx n/2$ and still have $|B_L| \approx |B|$.

\begin{lemma}\label{lem:halve}
  For sufficiently large $n$ and UDCPs $(A,B)$ such that $|A| \ge
  2^{(1-\epsilon)n}$, $|B| = 2^{\beta n}$, there exists $L \subseteq
  [n]$ such that $\frac{|L|}{n} \in \frac{1}{2} \pm \sqrt{\ln(2)
    \epsilon/2}$ and $|B_{L}| \geq 2^{(\beta - \epsilon)n-1}$.
\end{lemma}

\begin{proof}
Let $P \subseteq A \times B$ consist of all pairs $(a,b)$ such that $|a \symdiff b|\in \big(\tfrac{1}{2} \pm \sqrt{\ln(2) \epsilon / 2}\,\big)n$. We have that
\begin{align*}
	|P| &= \sum_{b \in B} \big|\big\{a \in A : |a\symdiff b| \in \big(\tfrac{1}{2} \pm \sqrt{\ln(2) \epsilon / 2}\,\big)n \big\}\big|,\\
			&\geq \sum_{b \in B} |A|/2 \;\;=\;\; |A|\cdot|B|/2,
\end{align*}
where the inequality is by Observation~\ref{obs:fatlayer}. Similarly as in the proof of Lemma~\ref{lem:vantilborg}, consider the encoding 
\[
 \eta: (a,b) \mapsto (a \symdiff b, b \setminus a).
\]
By Observation~\ref{obs:minudcp}, $|A - B|=|A|\cdot |B|$, and since $a-b$ can be computed from $\eta(a,b)$, it follows that $\eta$ is injective and $|\eta(P)|=|P|$.   We now upper bound $|\eta(P)|$. To this end, note that $b \setminus a \subseteq a \symdiff b$, and so $b\setminus a \in B_{a\symdiff b}$.\footnote{More precisely, $b \setminus a$ projected to $a \symdiff b$ is in $B_{a\symdiff b}$; we only need that $b \setminus a$ can be described by a single element of $B_{a\symdiff b}$.} Therefore, by summing over the possible values of $X= a \symdiff b$ we have that
\[
	|\eta(P)| \leq \sum_{\substack{X \subseteq [n] \\ |X| \in \left(\tfrac{1}{2} \pm \sqrt{\ln(2)\epsilon/2}\,\right)n}}|B_{X}|.
\]
This means that there must be an $X \subseteq [n]$ with $|X| \in \big(\tfrac{1}{2} \pm \sqrt{\ln(2)\epsilon/2}\,\big)n$ and $|B_X| \ge |\eta(P)| / 2^n = |P|/2^n \ge |A| \cdot |B| /2 / 2^n \ge 2^{(\beta - \epsilon)n - 1}$.
\end{proof}

\section{Combining the Bounds}

In this section we show how Lemmata~\ref{lem:encoding-bound},
\ref{lem:lowbnd}, and \ref{lem:halve} combine to yield our main theorem.

\begin{reptheorem}{thm:main}
  If $A,B \subseteq \{0,1\}^n$ is a UDCP with $|A| \geq
  2^{(1-\epsilon)n}$ and $|B| = 2^{\beta n}$, then $\beta \leq 0.4228
  +\sqrt{\epsilon}$.
\end{reptheorem}

\begin{proof}
  Without loss of generality, we may assume that $n$ is sufficiently
  large for all estimates to hold, since a lower bound for large $n$
  also holds for small $n$: if $(A_1,B_1)$ and $(A_2,B_2)$ are UDCPs,
  then so is $(A_1\times A_2, B_1\times B_2)$.

  By Lemma~\ref{lem:halve}, there exists a partition $L,R$ of $[n]$
  such that $\lambda = |L|/n \in \frac{1}{2} \pm \sqrt{\ln(2)
    \epsilon/2}$ and $2^{\pi n} := |B_{L}| \ge 2^{(\beta -
    \epsilon)n - 1}$.  By Claim~\ref{claim:freq}, there is an $A'
  \subseteq A$ such that $A$ is $\epsilon$-dense with respect to $L$.

  Applying Lemmata~\ref{lem:encoding-bound} and~\ref{lem:lowbnd} to
  the UDCP $(A', B)$ we then obtain that
\begin{align*}
  \frac{\pi-\lambda-\epsilon-2\rho\sqrt{\epsilon(\lambda-\pi)}}{1-\rho^2}+\lambda-1-\epsilon - o(1) &\le \frac{\log_2 \Pr_{a \sim^L_{\rho} b}[a \in A', b \in B]}{n} \\
  &\le \sqrt{\tfrac{\ln(2)\epsilon}{2}}-\tfrac{1}{2}+\lambda \cdot \left(\log_2(3-\rho)-\tfrac{3}{2}\right) + o(1).
\end{align*}

Simplifying, we get
\begin{align}
\label{eq:pi_bound}
\pi &\leq \left(\sqrt{\tfrac{\ln(2)\epsilon}{2}}+\tfrac{1}{2}+ \epsilon + \lambda \cdot \left(\log_2(3-\rho)-\tfrac{5}{2} \right) \right) \! (1-\rho^2) \nonumber \\
  &\hspace{20pt}+ 2\rho\sqrt{\epsilon(\lambda-\pi)}+\epsilon+\lambda + o(1).
\end{align}
We now set $\rho = 0.654$.  Plugging in this value and simplifying, \eqref{eq:pi_bound} becomes
\begin{align*}
  \pi & \le 0.2861421 + 0.2733156 \lambda + 1.573 \epsilon + 0.33691 \sqrt{\epsilon} + 1.308 \sqrt{\epsilon(\lambda - \pi)} + o(1).
\end{align*}
Using $\lambda \le \tfrac{1}{2} + \sqrt{\ln(2) \epsilon /2}$ and simplifying further, we get
\begin{align}
  \label{eq:pi_bound2}
  \pi & < 0.4228 + 1.573 \epsilon + \left(0.4979 + 1.3080 \sqrt{0.5 + \sqrt{\ln(2)\epsilon/2} - \pi}\right)\sqrt{\epsilon} + o(1).
\end{align}
Since $\beta \le \pi + \epsilon + o(1)$, we would like to show that
$\pi < 0.4228 + \sqrt{\epsilon} - \epsilon$.
Assume for the sake of contradiction that $\pi \ge 0.4228 +
\sqrt{\epsilon} - \epsilon$.  Plugging this into \eqref{eq:pi_bound2} gives
\begin{align}
  \label{eq:pi_bound3}
  0 &< 2.573 \epsilon + 
  \left(0.4979 - 1 + 1.308 \sqrt{0.0772 + \sqrt{\ln(2)\epsilon/2} - \sqrt{\epsilon} - \epsilon}\right) \sqrt{\epsilon} + o(1).
\end{align}
For $0 \le \epsilon \le 0.01$, it can be verified using a computer
that the right hand side of \eqref{eq:pi_bound3} is non-positive,
yielding the desired contradiction (for sufficiently large $n$), and
proving that $\beta < 0.4228 + \sqrt{\epsilon}$.  For $\epsilon >
0.01$, we have $\beta < 0.5 + \epsilon < 0.4228 + \sqrt{\epsilon}$
(the first inequality being the classic $|B| \le 2^{1.5n} / |A|$
upper bound).
\end{proof}

\section{Upper Bound Proof}
\label{sec:ub-proof}

In this section, we prove the upper bound on the refined noise probability stated in Lemma~\ref{lem:encoding-bound}.

\begin{replemma}{lem:encoding-bound}
Fix $L \subseteq [n]$ and let $\lambda = |L|/n$.  Then for any $0 \le \rho \le 1$ and UDCP
$(A,B)$ such that $|A|$ is $\epsilon$-dense with respect to $L$, we have
\[
\frac{\log_2 \Pr_{a \sim^L_{\rho} b}[a \in A, b \in B]}{n} \leq \sqrt{\tfrac{\ln(2)\epsilon}{2}}-\tfrac{1}{2}+\lambda \cdot \left(\log_2(3-\rho)-\tfrac{3}{2}\right) + o(1).
\]
\end{replemma}

\begin{proof}
Let $R = [n] \setminus L$ be the coordinates not in $L$.
Let $W_{d}$ be the set of pairs $a_{L}a_{R} \in A, b_{L}b_{R} \in B$ such that $|a_{L} \symdiff b_{L}|=d$. 

\begin{claim}\label{clm:wd}
For sufficiently large $n$, we have that $|W_d| \leq \binom{|L|}{d}2^{d}2^{1.5|R|}2^{n\sqrt{\ln(2)\epsilon/2}+1}$.
\end{claim}

\begin{proof}
Let $\epsilon'=\sqrt{\frac{\ln(2)\epsilon}{2(1-\lambda)}}$, and let $W'_d \subseteq W_d$ be all pairs from $W_d$ such that $\tfrac{|a_{R} \symdiff b_{R}|}{|R|} \in \tfrac{1}{2} \pm \epsilon'$.
Similarly as in the proof of Lemma~\ref{lem:halve}, we see that
\begin{align*}
	 |W'_d| & = \sum_{\substack{b_Lb_R \in B \\a_L \in A_{L} \\ |a_L \symdiff b_L|=d}}\!\! \left|\left\{a_R \in \{0,1\}^R: a_La_R \in A, |a_R \symdiff b_R|\in (\tfrac{1}{2} \pm\epsilon')|R|\right\}\right|,\\
	 &\ge \sum_{\substack{b_Lb_R \in B \\a_L \in A_{L} \\ |a_L \symdiff b_L|=d}} \tfrac{1}{2} |\{a_R \in \{0,1\}^R: a_La_R \in A\}| \;\;=\;\; \tfrac{1}{2} |W_d|.
\end{align*}
The inequality follows from Observation~\ref{obs:fatlayer} combined with the $\epsilon$-dense property $|\{a_R \in \{0,1\}^R: a_La_R \in A\}| \geq 2^{|R|-\epsilon n}/2=2^{(1-\epsilon/(1-\lambda))|R|}/2$.

We proceed with upper bounding $|W'_d|$. Similarly as in the proof of Lemma~\ref{lem:vantilborg}, we define an encoding $\eta$ on elements $(a,b)$ of $W'_d$: 
\[
 \eta: (a_La_R,b_Lb_R) \mapsto (a_L \symdiff b_L, a_L \setminus b_L,a_R \symdiff b_R,a_R \setminus b_R)\,.
\]
Since the image $\eta(a,b)$ directly gives $a-b$ and we know that $|A-B|=|A||B|$ by Observation~\ref{obs:minudcp}, we have that $\eta$ is injective and thus
\[
|W'_d| = |\eta(W'_d)| \leq \binom{|L|}{d}2^{d}\sum_{i \in (0.5\pm\epsilon')|R|} \binom{|R|}{i} 2^{i},
\]
where the inequality follows by bounding the number of possibilities in every coordinate of $\eta(\cdot)$. The claim is then implied for sufficiently large $n$ from the easy observation that 
\[
	\sum_{i \in (0.5 \pm \epsilon')|R|} \binom{|R|}{i} 2^{i} \leq 2^{(1.5+\epsilon')|R|}\leq 2^{1.5|R|+n\sqrt{\ln(2)\epsilon/2}}\,.
\]
\end{proof}

By the refined definition of $\sim^L_\rho$ we have that
\begin{equation}\label{eq:refrho}
  \Pr_{a \sim^L_{\rho} b}[a \in A, b \in B] = 2^{-n}\sum_{d=0}^{|L|} \Big{(}\frac{1+\rho}{2}\Big{)}^{|L|-d} \Big{(}\frac{1-\rho}{2}\Big{)}^{d}2^{-|R|} W_{d}\,.
\end{equation}
To see that this is true, note that $W_d$ counts exactly the pairs $a \in A, b\in B$, such that $|a_L \symdiff b_L|=d$, and that the probability that such pair is picked can be computed as the probability that $a$ is picked (which is $2^{-n}$), times the probability that $b$ is picked given that $a$ is picked. The probability that $b_R$ is picked is simply $2^{-|R|}$ since it is picked uniformly at random, and the probability that $b_L$ is picked is $\Big{(}\frac{1+\rho}{2}\Big{)}^{|L|-d} \Big{(}\frac{1-\rho}{2}\Big{)}^{d}$, similarly as in the proof of Theorem~\ref{thm:simplebound}.

Using Claim~\ref{clm:wd}, we upper bound~\eqref{eq:refrho} by
\begin{align*}
  \Pr_{a \sim^L_{\rho} b}[a \in A, b \in B] & \le 2^{-2n}\sum_{d=0}^{|L|} (1+\rho)^{|L|-d} (1-\rho)^{d} \binom{|L|}{d}2^{d}2^{1.5|R|+n\sqrt{\ln(2)\epsilon/2}+1}\\
                 &= 2^{-2n+1.5|R|+n\sqrt{\ln(2)\epsilon/2} + 1}\sum_{d=0}^{|L|} (1+\rho)^{|L|-d} (2-2\rho)^{d}\binom{|L|}{d}\\
                 &= 2^{\left(\sqrt{\ln(2)\epsilon/2}-2\right)n+1.5|R| + 1}(3-\rho)^{|L|},
\end{align*}
where the last equality follows from the Binomial Theorem.
Using $|R| = n-|L|$, taking logs, and dividing by $n$, we get
\[
\frac{\log_2 \Pr_{a \sim^L_{\rho} b}[a \in A, b \in B]}{n} \le
 \sqrt{\tfrac{\ln(2)\epsilon}{2}}-\tfrac{1}{2}+\lambda\left(\log_2(3-\rho)-\tfrac{3}{2}\right) + 1/n.
\]
\end{proof}

\section{Lower Bound Proof}
\label{sec:lb-proof}

In this section, we prove the lower bound on the refined noise probability.

\begin{replemma}{lem:lowbnd}
  Fix $L \subseteq [n]$ with $|L| = \lambda n$.  Then for any constant
  $0 \le \rho < 1$ the following holds.  Let $(A,B)$ be a UDCP such
  that $A$ is $\epsilon$-dense with respect to $L$, and $|B_L| =
  2^{\pi n}$ for some $0 \le \pi \le \lambda$. Then
  \[
  \frac{\log_2 \Pr_{a \sim^L_{\rho} b}[a \in A, b \in B]}{n} \ge \frac{\pi-\lambda-\epsilon-2\rho\sqrt{\epsilon(\lambda-\pi)}}{1-\rho^2}+\lambda-1-\epsilon - o(1).
  \]
  The constant in the $o(1)$ term depends on $\lambda, \rho,
  \epsilon$ and $\pi$, and is finite assuming $\epsilon$ is bounded
  away from $0$ and $\rho$ is bounded away from $1$.
\end{replemma}
\begin{proof}
Due to the chain rule, $\Pr_{a \sim^L_{\rho} b}[a \in A, b \in B]$ equals
\begin{equation}\label{eq:chainrule}
	 \Pr_{a \sim^L_{\rho} b}[a \in A, b \in B \,|\, a_{L} \in A_{L}, b_L \in B_{L} ]\,
	\cdot\  \Pr_{a_L  \sim_{\rho} b_L}[a_{L} \in A_{L}, b_L \in B_{L}]\,.
\end{equation}
We proceed with lower bounding the first term of~\eqref{eq:chainrule}.  Let $R = [n] \setminus L$.  For the first factor, note that if $b_{L} \in B_{L}$, there is at least one $b_{R}$ such that $b_{L}b_{R}\in B$ by the definition of $B_{L}$, and such a $b_{R}$ is picked with probability $2^{-|R|}$ since it is uniformly distributed over $2^R$.  Similarly, if $a_{L} \in A_{L}$, there are at least $2^{|R|-\epsilon n}/2$ sets $a_{R}\subseteq R$ such that $a_{L}a_{R}\in A'$ by the definition of $A'$, and so such an $a_{R}$ is picked with probability at least $2^{-\epsilon n}/2$. In summary, we have that
\[
	\Pr_{a \sim^L_{\rho} b}[a \in A, b \in B \,|\, a_{L} \in A_{L}, b_L \in B_{L} ] \geq 2^{-|R|-\epsilon n}/2 = 2^{(\lambda - 1 - \epsilon - o(1))n}.
\]
For the second term, apply Theorem~\ref{thm:rsse} with $U = L$ and
\begin{align*}
  F&=A_{L}, & f &= \frac{|L|-\epsilon n - 1}{|L|} = 1 - \frac{\epsilon}{\lambda} - o(1), \\
  G&= B_{L}, & g &= \frac{\pi}{\lambda}, \\
\end{align*}
which gives that 
\begin{align*}
\log_2 \Pr_{a_L  \sim_{\rho} b_L}[a_{L} \in A_{L}, b_L \in B_{L}]
	 &\geq -|L|\left(\frac{(1-\tfrac{\pi}{\lambda})+\tfrac{\epsilon}{\lambda}+o(1)+2\rho\sqrt{(1-\tfrac{\pi}{\lambda})(\tfrac{\epsilon}{\lambda} + o(1))}}{1-\rho^2}\right),\\
	&= n\left(\frac{\pi - \lambda - \epsilon - 2\rho\sqrt{\epsilon\lambda-\epsilon\pi}}{1-\rho^2} - o(1)\right).
\end{align*}
The statement now follows by multiplying the two lower bounds following~\eqref{eq:chainrule}.
%\begin{align*}
%	& 2^{-n\left(\frac{\lambda-\pi+6\sqrt{\epsilon}}{1-\rho^2}\right)} 2^{-|R|-2\epsilon n} \\
%	&=2^{n\left(\frac{\pi-\lambda-6\sqrt{\epsilon}}{1-\rho^2}+\lambda-1-2\epsilon\right)}.
%\end{align*}
\end{proof}

\paragraph*{Acknowledgements}

This research was funded by
the Swedish Research Council, Grant 621-2012-4546 (PA), 
the European Research Council, Starting Grant 338077 
``Theory and Practice of Advanced Search and Enumeration'' (PK),
the Academy of Finland, Grant 276864 ``Supple Exponential Algorithms'' (MK), 
and NWO VENI project 639.021.438 (JN).

\bibliographystyle{IEEEtranS}
\bibliography{udcp}

\end{document}